\documentclass[a4paper,11pt]{article}
\usepackage[utf8]{inputenc}
\usepackage[T1]{fontenc}
\usepackage[centertags]{amsmath}
\usepackage{amsfonts}
\usepackage{amssymb}
\usepackage{amsthm}
\usepackage{fancyhdr}
\usepackage[left=2.5cm,top=3.0cm,right=2.5cm,bottom=2cm,headheight=3.0cm]{geometry}
\usepackage[onehalfspacing]{setspace}
\usepackage{parskip}
\usepackage{natbib}

\setcitestyle{authoryear,open={(},close={)}}

\usepackage[dvipsnames]{xcolor}
\usepackage[unicode=true, colorlinks = true, allcolors=MidnightBlue]{hyperref}

\usepackage{tikz}
\usetikzlibrary{calc}

\renewcommand*{\thefootnote}{\fnsymbol{footnote}}

\usepackage [english]{babel}
\usepackage [autostyle, english = american]{csquotes}
\MakeOuterQuote{"}

\newcommand{\E}{\mathbb{E}}
\newcommand{\Prob}{\mathbb{P}}
\newcommand{\R}{\mathbb{R}}

\newcommand{\dd}{\mathrm{d}}

\renewcommand{\le}{\leqslant}
\renewcommand{\leq}{\leqslant}
\renewcommand{\ge}{\geqslant}

\newcommand{\defin}[1]{\textbf{#1}}

\DeclareMathOperator*{\argmax}{arg\,max}
\DeclareMathOperator*{\argmin}{arg\,min}

\usepackage[inline]{enumitem}
\usepackage{mleftright}

\theoremstyle{definition}
\newtheorem{theorem}{Theorem}

\newtheorem{lemma}{Lemma}
\newtheorem{remark}{Remark}

\begin{document}

\pagenumbering{gobble}

\begin{center}
\LARGE{\textbf{ Convolution Mode Regression}}
\\
\vspace{1.5cm}          
\href{https://drive.google.com/file/d/1g_ZGEUWOutLVQHKLoTQPZguehHIt2eIQ/view?usp=sharing}{Eduardo Schirmer Finn}\footnote[4]{Corresponding Author, MSc in Statistics at the Federal University of Rio Grande do Sul (UFRGS), Brazil. \href{mailto:eduardosfinn@outlook.com}{eduardosfinn@outlook.com}}, 
\href{https://professor.ufrgs.br/eduardohorta/}{Eduardo Horta}\footnote[2]{Adjoint Professor in the Department of Statistics at the Federal University of Rio Grande do Sul (UFRGS), Brazil. \href{mailto:eduardo.horta@ufrgs.br}{eduardo.horta@ufrgs.br}}
\vspace{1.5cm}
\end{center}

\vspace{1cm}
\large{\textbf{\textcolor{MidnightBlue}{Abstract}}}
\vspace{0.25cm}

For highly skewed or fat-tailed distributions, mean or median-based methods often fail to capture the central tendencies in the data. Despite being a viable alternative, estimating the conditional mode given certain covariates (or mode regression) presents significant challenges. Nonparametric approaches suffer from the "curse of dimensionality", while semiparametric strategies often lead to non-convex optimization problems. In order to avoid these issues, we propose a novel mode regression estimator that relies on an intermediate step of inverting the conditional quantile density. In contrast to  existing approaches, we employ a convolution-type smoothed variant of the quantile regression. Our estimator converges uniformly over the design points of the covariates and, unlike previous quantile-based mode regressions, is uniform with respect to the smoothing bandwidth. Additionally, the \textit{Convolution Mode Regression} is dimension-free, carries no issues regarding optimization and preliminary simulations suggest the estimator is normally distributed in finite samples.

\vspace{0.25cm}

\textit{\textbf{\textcolor{MidnightBlue}{Keywords:}} Mode Regression; Convolution-based Smoothing; Conditional Quantile;  Asymptotic Theory; Uniform Convergence.}

\newpage

\tableofcontents

\newpage

\renewcommand{\thefootnote}{\arabic{footnote}}

\pagenumbering{arabic}
\setcounter{page}{3}

\section{Introduction}

Conventional econometric methods are generally mean-(or median)-based; such methods may fail to express the central tendency if distributions are highly skewed or long-tailed  \citep{Kemp_Silva_2012, Chen_et_al_2016}. The conditional mode emerges as a robust alternative, conveying the desirable interpretation of being the most likely value of a dataset \citep{chacon_2020}. When data is discrete, the conditional sample mode is more straightforward, however, in a continuous variable setting, it can be estimated as the point that maximizes the conditional probability density function (pdf). Since \citet{Lee_1989}, the estimation of conditional mode, called mode regression, has demonstrated its utility across various domains, specially in applications with asymmetric data, such as wages \citep*{Zhang_Kato_Ruppert_2021}; electrical energy consumption \citep*{ota_kato_hara_2019}; medical sciences \citep{wang_et_al_2017}, traffic data \citep{einbeck_tutz_2006} and wildfire burned areas \citep{yao_li_2014}.

In reviewing the conditional mode literature, two fundamental considerations emerge in the discussion: firstly, the assumption of whether the mode is global (unique) or local; secondly, if the estimation strategy employed is semiparametric (linear) or fully nonparametric. Linear approaches require the condition that the mode is a unique maximizer of the conditional density, while nonparametric techniques in turn are usually used for multimodal models.\footnote{Nonparametric estimation is not exclusive to multimodal specifications, since in few cases the global mode is also estimated nonparametrically \citep{sager_thisted_1982, feng_et_al_2020}.} The first semiparametric estimator considering a unique mode was developed by \citet{Lee_1989} and establishes a linear relationship between the mode of the response and its covariates; however, despite being elegant, this model is impractical;\footnote{Has restrictive assumptions on the conditional density of the response and, due to the objective function, the estimator lacks a tractable distribution \citep{Kemp_Silva_2012}.} subsequent models, such as the ones of \citet{Kemp_Silva_2012} and \citet{yao_li_2014}, yield out non-convex optimization problems, resulting in functions that may have multiple maxima; also, the related algorithms are sensible to the starting points. On the other hand, although nonparametric estimation tends to avoid misspecification \citep{yao_et_al_2012, Chen_et_al_2016}, these methods suffer from slow convergence even in moderate dimension sizes. 

Two recent developments \citep*{ota_kato_hara_2019, Zhang_Kato_Ruppert_2021} have addressed the optimization and dimension issues by estimating the mode via a quantile regression (QR). This strategy relies on retrieving the conditional density (whose maximizer is the mode) from the derivative of the conditional quantile function, referred to as the conditional quantile density function (qdf). Similarly to what is done in \citet*{ota_kato_hara_2019} and \citet*{Zhang_Kato_Ruppert_2021}, we propose a novel approach which relies on working with the (convolution-based) smoothed version of the QR estimator of \citet*{fernandes_guerre_horta_2021}. This alternative provides some advantages: it is continuous, unbiased and less variable than the traditional QR estimator.

The main goal of this work is to derive the asymptotic properties for the estimator of the conditional mode via smoothed quantile regression, hereby denominated \textit{Convolution Mode Regression}. Our procedure differs from \citet*{Zhang_Kato_Ruppert_2021}, as they opt for an \textit{"estimate then smooth"} process, using the traditional QR framework;\footnote{Parallel to \citeauthor{parzen_1979}'s \citeyear{parzen_1979} proposition of quantile estimation by smoothing the sample quantile function.} whereas, we prefer to \textit{"smooth then estimate"}, stemming from the smoothed quantile regression estimator.\footnote{Resembling the method of \citet{Nadaraya_1964}, in which the unconditional quantile is estimated via inverting a smoothed estimator of the cdf (cumulative distribution function).} The convergence rate of the \citet*{Zhang_Kato_Ruppert_2021} estimator is $O_P( n^{-1/2}h^{-3/2} \sqrt{\log n} + h^2)$, whilst ours is $O_P( (\frac{\log n} {nh})^{1/4}) + o(h^{1/2})$; both rates dimension-free. Despite \citeauthor*{Zhang_Kato_Ruppert_2021}'s rate, their convergence is uniform solely in the design points of the covariates. In contrast to our estimator, which holds uniformity with respect to the design points, but also for the smoothing bandwidth. In light of this, our estimator can be more adequate in cases where the choice of the bandwidth is data-driven. 

\subsection{Literature Review}
Following the introduction of global mode estimators by \citet{parzen_1962} and \citet{chernoff_1964},\footnote{\citet{sager_1978} categorized estimators as: "naive" when the density was not estimated (Chernoff), and as "indirect", since an intermediate step, such as Kernel Density Estimation, is required (Parzen).  These classifications may blur, as the majority of direct estimators reveal some form of linkage with a type of density estimation \citep{sager_1978, chacon_2020}.} \citet{sager_thisted_1982} generalized the framework of the latter and developed the first mode regression, establishing the global mode of the dependent variable as a monotone function of the covariate. Albeit its application only to ordinal data, this model showed that the conditional mode estimator could be formulated by applying a plug-in from a density estimator, such as a Kernel Density Estimator (KDE). Subsequently, mode regression literature has evolved and can be labeled according to two major factors: firstly, the assumptions of the model, as in a unique global mode vs. multiple local modes; secondly, the type of estimation,  semiparametric (which postulates a linear relationship between response variable and covariates) vs. nonparametric (which allows for multiple local modes). Semiparametric estimation only works with the assumption of the mode being unique, however, a unique mode can be estimated nonparametrically. In light of this, we can divide the mode regression literature into 4 different strands:\footnote{Some papers may not fit precisely in this categorization, since they mix parametric and nonparametric traits \citep{liu_etal_2013, wang_2024}, or use a Bayesian approach \citep{yu_arist_2012}.}
(1) Unique Mode with Semiparametric Estimation; (2) Unique Mode with Nonparametric Estimation, and (3) Multimodal with Nonparametric Estimation; afterward, special attention is given to (4) Conditional Quantile Approaches towards the mode, since this is more related to our contribution.\footnote{For a general survey of the role of the mode in statistics, we recommend \citet{chacon_2020}; especially for mode regression, we indicate the review of \citet{chen_2018}.}

\textbf{(1) Unique Mode \& Semiparametric Estimation:} considering a linear framework, \citeauthor{Lee_1989} (\citeyear{Lee_1989, Lee_1993}) proposed to use a kernel smoothed loss function to estimate the conditional mode of truncated variables. The main drawback of this line of action lies in the underlying assumptions on the conditional density of the dependent variable, as well as homogeneity and symmetry of the error terms, which led to the conditional mode coinciding with the conditional mean \citep{chen_2018}.  Further work that sought to eliminate the error's assumptions, such as: \citet{Kemp_Silva_2012}, where the loss function allows for smooth kernels; as well as  \citet{yao_li_2014}, who focused on high-dimensional data. Both methods have algorithmic issues, leading to a nonconvex optimization problems, with no guarantee of convergence to the global maximum, and high sensitivity to the selected starting point.\footnote{Some further exploration of this literature strand can be found in variable selection \citep{zhang_2013}, time series analysis \citep{kemp_etal_2020} and in panel data \citep{ullah_etal_2021}.}

\textbf{(2) Unique Mode \& Nonparametric Estimation:} linearity can be too restrictive depending on the type of data; thus, nonparametric regression can model the components of the conditional mode as smooth functions of the covariates \citep{Chen_et_al_2016}. A generalization of local polynomial regression is provided by \citet{yao_et_al_2012}, denoted local modal regression, which presents better efficiency when errors are heavy-tailed. Nonetheless, the model suffers from the “curse of dimensionality” and error's symmetry is imposed.\footnote{According to \citet*{Zhang_Kato_Ruppert_2021}, their sixth assumption (symmetry of the error term) leads to the problem corresponding to conditional mean estimation.} A more recent nonparametric regression for single mode that mitigates the “curse of dimensionality” is found in \citet{feng_et_al_2020}. This estimation of the conditional mode is achieved via an empirical risk minimization approach, which turns the problem into non-dependable on dimension.

\textbf{(3) Multimodal \& Nonparametric Estimation:} as it is not always the case that data structures can be interpreted as unimodal, \citet{scott_1992} proposed to consider various modes as points of local maxima of the conditional density of the response. Both \citet{matzner_etal_1998} and \citet{einbeck_tutz_2006} solved prediction problems using the conditional mode. Using less restrictive assumption on the kernel density function, \citet{Chen_et_al_2016} developed a KDE-based mode regression with strong asymptotic properties and no specification issues. This model is more general than ours, since it allows for multiple modes, while we assume a unique mode at each design point of interest. Despite this, the model is susceptible to the "curse of dimensionality", given that the convergence is slow even with a modest number of regressors \citep*{Zhang_Kato_Ruppert_2021}.

\textbf{(4) Conditional Quantile Approach:} motivated by the challenges of nonconvex optimization and misspecification in linear models, as well as the susceptibility of nonparametric estimators to the "curse of dimensionality", \citet*{ota_kato_hara_2019} developed a novel semiparametric approach with an intermediate step based on quantile regression. The traditional quantile function of \citet{Koenker_Bassett_1978} is estimated, following,  a conditional \textit{quantile density} estimator is obtained via numerical differentiation. The underlying model does not impose linearity in the mode function, even when the quantile regression model is linear-in-covariates. Furthermore, it is computationally appealing, since quantile regression can be written as a linear programming problem. Still, the \citeauthor{Koenker_Bassett_1978}'s (\citeyear{Koenker_Bassett_1978}) empirical conditional quantile function is not continuous---hence the mentioned numerical differentiation, which leads to an approximate optimizer for the mode. In order to surpass this problem, \citet*{Zhang_Kato_Ruppert_2021} propose to post-smooth the quantile regression estimator by a kernel function. This strategy not only circumvents the approximate optimization, but also yields faster convergence rates. Both models take off from a key identity that we also explore in this paper, namely, that the quantile density is the reciprocal of the density function, evaluated at the quantile of interest, which summarizes how the conditional mode can be retrieved from the quantile density. 

\subsection{Organization}
The rest of this paper is organized as follows. In Section \ref{Model_Section} the model is presented, we introduce the estimator and explore its relationship with the smoothed quantile regression. In Section \ref{Main Results Section} we enunciate the main mathematical results of the paper, as well as commentary regarding convergence rates and different bandwidth scenarios.  Section \ref{Simulation} is a preliminary Monte Carlo study in which we simulate the possible limiting distribution for the \textit{Convolution Mode Regression Estimator}. Finally, in Section \ref{Concluding Remarks} we state our concluding remarks along with possibilities for future work. The Appendix contains additional mathematical material, such as derivations and convergence rates calculations.


\section{Convolution Mode Regression} \label{Model_Section}
\subsection{Setup}

Let $Y \in \mathbb{R}$ represent a target random variable for which we are interested in estimating the conditional mode, given a $d$-dimensional vector $X$ of covariates, and write $\mathcal{X}:=\mathrm{support}(X)$. Assume that $Y|X=x$ is continuous and unimodal, having conditional cdf $F(\cdot|x)$ and conditional pdf $f(\cdot|x)$. Then, the \defin{conditional mode} of $Y$ given $X=x$, denoted by $m(x)$, is defined as:
\begin{equation} \label{mode_def}
    m(x) = \argmax_{y \in \mathbb{R}} f(y|x),\qquad x\in \mathcal{X}.
\end{equation}
Thus, $m(x)$ corresponds to the point in the covariate space at which the (conditional) density of the response attains its maximum value. Additionally, define the \defin{$\tau$-th conditional quantile of $Y$ given $X=x$} as the scalar $Q(\tau|x)$ given by
\begin{equation*}
    Q(\tau|x) := \inf \lbrace y \in \mathbb{R}: F(y|x) \ge \tau\rbrace, \quad \tau \in (0,1), x \in \mathcal{X}
\end{equation*}
and the \defin{conditional quantile function} as the mapping $\tau \mapsto Q(\tau|x)$.

It is important to point out that the quantile function is entirely retrievable from the cdf, since it is just the generalized inverse of the function $y\mapsto F(y|x)$, and in the case of the distribution being a continuous function we have that $Q(\cdot|x):=F^{-1}(\cdot|x)$ for each $x$ \citep{vdVaart_1998, koenker_2005}. Furthermore, the \defin{conditional quantile density function (qdf)} is defined through
\begin{equation}
q(\tau |x) = Q'(\tau|x) = \frac{\partial Q(\tau|x)}{\partial \tau},\quad \tau\in(0,1),\,x\in\mathcal{X}.
\end{equation}

A key identity explored by \citet*{ota_kato_hara_2019} and \citet*{Zhang_Kato_Ruppert_2021} is that, as a consequence of the Inverse Function Theorem, the identity
\begin{equation} \label{key_id}
    q(\tau|x) = \frac{1}{f(Q(\tau|x)|x)}
\end{equation}
holds for every allowable $x$ and $\tau$. In this sense, given some regularity conditions which we introduce below, we can minimize the inverse of the density (which is the qdf), as in equation \eqref{key_id}, and retrieve the maximizer of $y \mapsto f(y|x)$ from equation \eqref{mode_def}; thus,
\begin{equation}
    m(x) = Q\big(\argmin\nolimits_{\tau} q(\tau|x)\,|\,x\big)
\end{equation}

Regarding the quantile function, both \citet*{ota_kato_hara_2019} and \citet*{Zhang_Kato_Ruppert_2021} consider the quantile regression model developed by \citet{Koenker_Bassett_1978}, which stipulates a linear-in-covariates representation of the conditional quantile function:
\begin{equation} \label{quantile}
    Q(\tau|x) = x^\intercal \beta(\tau), \quad \tau \in (0,1),\, x \in \mathcal{X},
\end{equation}
where $\beta : (0,1) \mapsto \mathbb{R}^d$ is a functional parameter. For each fixed $\tau$ in the interval $(0,1)$, the vector $\beta(\tau)$ in \eqref{quantile} solves a similar minimization problem as the one found in classic linear regression. For this end, the following population objective function is proposed:
\begin{equation} \label{Q_pop_obj_function}
    R(b; \tau) := \E [ \rho_\tau (Y - X^\intercal b )] = \int \rho_\tau (t)\, \dd F(t;b)
\end{equation}
with $\rho_\tau (u) := u[\tau - \mathbb{I}(u<0)]$ known as the check function. The true parameter $\beta(\tau)$ minimizes $R(b;t)$ with respect to $b\in\R^d$. The sample equivalent proposed by \citet{Koenker_Bassett_1978} is defined as:
\begin{equation} \label{Q_sample_obj_function}
    \widehat{R}(b;\tau) := \frac{1}{n} \sum^{n}_{i=1} \rho_\tau (Y_i - X_i^\intercal b) = \int \rho_\tau (t)\, \dd \widehat{F}(t;b),
\end{equation}
where $\widehat{F}(\cdot;b)$ is the empirical distribution function of $\varepsilon_i(b):=Y_i - X_i^\intercal b$, for $i=1,...,n$, with the traditional quantile regression estimator as the minimizer of $\widehat{R}(b;\tau)$, with respect to $b\in\R^d$, that is:
\begin{equation} \label{traditional_QRE}
\widehat{\beta}(\tau) = \argmin_{b\in\R^d} \widehat{R}(b;\tau),\qquad \tau\in(0,1) 
\end{equation}

According to Theorem 2.1 in \citet{Bassett_Koenker_1982}, the empirical conditional quantile function $\tau\mapsto x^\intercal \widehat{\beta}(\tau)$ exhibits jumps, in particular it is not differentiable. To overcome this issue, \citet*{fernandes_guerre_horta_2021} proposed using a kernel-type cdf estimator, similar to \citet{Nadaraya_1964}, instead of the empirical distribution function. The resulting smoothed version of the sample objective function in \eqref{Q_sample_obj_function} is:
\begin{equation} \label{smoothed_QR}
    \widehat{R}_h(b; \tau) = \frac{1}{n} \sum_{i=1}^n k_h \ast \rho_\tau (Y_i - X_i^\intercal b) = \int \rho_\tau (t) \widehat{f}_h(t;b)\,\dd t
\end{equation}
where the symbol $\ast$ denotes the convolution operator, and where $\widehat{f}_h(\cdot;b)$ is the kernel estimator of the density of $Y_i - X_i^\intercal b$. Here, $k_h(u) = k(u/h)/h$, where $k\colon \R\to\R_+$ is a smooth kernel function and $h>0$. The new estimator is the minimizer of the objective function \eqref{smoothed_QR}, called the \defin{smoothed quantile regression estimator (SQRE)} and defined by:
\begin{equation} \label{SQRE}
   \widehat{\beta}_h(\tau) := \argmin_{b \in \mathbb{R}^d} \widehat{R}_h(b; \tau),\qquad\tau\in(0,1).
\end{equation}

The mapping $ \tau \mapsto \widehat{\beta}_h(\tau)$ is continuously differentiable over the interval $(0, 1)$, unlike $\widehat{\beta}$. Differentiability offers notable advantages, and the reasons are twofold: (i) the smoothness of the objective function ensures the regularity of the resulting estimator; (ii) the asymptotic covariance matrix of $\widehat{\beta}_h(\tau)$ can be estimated in a standard fashion, as in \citet{Newey_McFadden_1994}. Regarding differentiability, writing $\widehat{R}^{(1)}_h(b;\tau):= {\partial \widehat{R}_h(b;\tau)}/{\partial b}$, the SQRE satisfies the first-order condition $\widehat{R}^{(1)}_h(\widehat{\beta}_h(\tau);\tau) = 0$. Accordingly, following the Implicit Function Theorem, we obtain:
\begin{equation} \label{beta_hat_1}
    \widehat{\beta}_h^{(1)}(\tau) = 
    \frac{\partial \widehat{\beta}_h(\tau)}{\partial \tau} :=  
    \Big[\widehat{R}_h^{(2)}\big(\widehat{\beta}_h(\tau); \tau\big)\Big]^{-1} \bar{X}
\end{equation}
Explicit formulas for the first and second order derivatives of $\widehat{R}_h(b;\tau)$ with respect to $b$ (respectively, $\widehat{R}^{(1)}_h(b;\tau)$ and $\widehat{R}^{(2)}_h(b;\tau)$) are provided in equation \eqref{derivates_of_Rh} in \ref{Appendix 0}. 


\subsection{Estimation}
Consider the following objective (or "sparsity"\footnote{This is the nomenclature used by \citet*{ota_kato_hara_2019} and \citet*{Zhang_Kato_Ruppert_2021}.}) function:
\begin{equation} \label{sparsity-pop}
    s_x(\tau):=
    -q(\tau|x), 
    \quad \tau\in(0,1),\,x\in\mathcal{X}.
\end{equation}
It is not difficult to show that
\begin{equation}\label{sparsity-pop-explicit}
s_x(\tau) = 
-x^\intercal {\beta^{(1)}(\tau)} = -x^\intercal [D(\tau)]^{-1} \mathbb{E}X 
\end{equation}
where
\begin{equation}\label{Ds_Rs-pop}
D(\tau) = R^{(2)}(\beta(\tau);\tau) = \mathbb{E}[XX^\intercal f(X^\intercal \beta(\tau)|X)].
\end{equation}

Under some regularity assumptions that will be introduced below, the function $\tau\mapsto s_x(\tau)$ has a unique maximizer, denoted $\tau_x$, which we call the \defin{conditional quantile mode} of $Y$ given $X=x$. If we plug in this optimizer in the quantile function $Q(\cdot|x)$ we get the expression $m(x) = Q(\tau_x|x)$. Consequently, the estimation of $m(x)$ boils down to estimating the conditional quantile function, and $\tau_x$. In view of \eqref{sparsity-pop-explicit}, we define, for conformable $\tau$, $x$ and $h$, the \defin{sample conditional sparsity function} as:
\begin{equation} \label{sparsity-sample-explicit}
    \widehat{s}_{x,h}(\tau) =
    -x^\intercal {\widehat{\beta}_h^{(1)}(\tau)} = - x^\intercal [\widehat{D}_h(\tau)]^{-1} \bar{X},
\end{equation}
where
\begin{equation} \label{Ds_Rs-sample}
\widehat{D}_h(\tau) = \widehat{R}^{(2)}_h(\widehat{\beta}_h(\tau);\tau) = \frac{1}{n} \sum_{i=1}^{n}X_iX_i^\intercal k_h(X_i^\intercal \widehat{\beta}_h(\tau)-Y_i),      
\end{equation}
see \citet*{fernandes_guerre_horta_2021}.

The optimizers for the sparsity functions, both population and sample, as defined in \eqref{sparsity-pop-explicit} and \eqref{sparsity-sample-explicit}, are given by:
    \begin{equation}\label{taus}
        \tau_x = \argmax_{\tau \in (0,1)} s_x(\tau)
        \quad \mathrm{and} \quad 
        \widehat{\tau}_{x,h} = \argmax_{\alpha\le\tau\le1-\alpha} \widehat{s}_{x,h}(\tau)
    \end{equation}
where $0<\alpha<1/2$ is a constant.

Our proposed \defin{smoothed conditional mode estimator} is then given by
\begin{equation} \label{mode_asQ}
    \widehat{m}_h(x) := \widehat{Q}_{x,h}(\widehat{\tau}_{x,h}) = x^\intercal \widehat{\beta}_h (\widehat{\tau}_{x,h}),
\end{equation}
for all $x\in\mathcal{X}$ and every allowable $h$.


\section{Main Results} \label{Main Results Section}
Before providing consistency results of the proposed estimator, we state the conditions for which our results are derived. 

\subsection{Assumptions:}

\begin{itemize}
\item \textbf{A1:}
The support of $X$, denoted $\mathcal{X}$, is compact and a subset of $\bar{\mathbb{R}}^{d}_{+ \ast}$, i.e.,\  the components of $X$ are positive, bounded RVs. The matrix $\E [XX^\intercal]$ is full rank. 
\item \textbf{A2:}
The mapping $\tau \mapsto \beta(\tau)$ is three times continuously differentiable.
\item \textbf{A3:}
The conditional density $f(y|x)$ is continuous and strictly positive over $\mathbb{R} \times \mathcal{X}$. Also, the derivative $f^{(1)}(\cdot|\cdot)$ exists and is uniformly continuous in the sense that 
\begin{equation*}
\lim_{\epsilon\to0} \sup_{(x,y)\in\mathbb{R}^{d+1}}\sup_{t:|t|\le \epsilon} \big\vert f^{(1)}(y+t|x) - f^{(1)}(y|x)\big\vert = 0, 
\end{equation*} 
and that $\sup_{(x,y)\in\mathbb{R}^{d+1}}\vert f^{(j)}(y|x)\vert < \infty$ and $\lim_{y\to\pm\infty} f^{(j)}(y|x)=0$ for all \(j \in \{0,1\}\).

\textbf{Remark}. The degree of differentiability of \(f(\cdot | \cdot)\) is used in \citet*{fernandes_guerre_horta_2021} to control the order of the smoothing kernel. Here, we set the maximum value of $j$ equal to 1, for simplicity.

\item \textbf{A4:}
The kernel $k\colon\mathbb{R} \mapsto \mathbb{R}$ is even, integrable and has bounded first and second derivatives. Additionally, $\int k(z) \dd z=1$; $0<\int^{\infty}_0 K(z)[1-K(z)]\dd z < \infty$ and, lastly, $0 < \int z^2k(z) \dd z < \infty$.

\item \textbf{A5:}
$h \in [\underline{h}_n, \bar{h}_n]$ with $n\underline{h}_{n}^3 /\log n\to\infty$ and $\bar{h}_n = o(1)$.

\item \textbf{A6:}
For all $x\in \mathcal{X}$, there exists $\tau_x\in(0,1)$ such that, for every $\epsilon>0$, it holds that
\begin{equation*}
    \sup_{\tau : |\tau-\tau_x| \ge \epsilon} s_x(\tau) < s_x(\tau_x)
\end{equation*}

\item \textbf{A7:}
For some $0<\alpha<1/2$, it holds that
\begin{equation*}
    \alpha < 
    \inf_{x\in\mathcal{X}} \tau_x \leq
    \sup_{x\in\mathcal{X}} \tau_x <
    1-\alpha.
\end{equation*} 
\end{itemize}

Assumptions A1-A5 are taken directly from \citet*{fernandes_guerre_horta_2021}, with minor modifications. Due to A1-A3, the Hessian $D(\tau)$ as defined in \eqref{Ds_Rs-pop}, is positive definite for all possible values of $\tau \in (0,1)$, therefore, $D(\tau)$ is invertible. Additionally, A2 ensures the function $\tau\mapsto Q(\tau|x)$ is increasing over the interval $(0,1)$, and, together with A3, that its derivative with respect to $\tau$ is  strictly positive. Also, A3 expresses some ordinary regularity conditions which guarantee smoothness of $f(\cdot | \cdot)$ \citep{koenker_2005}. Similar conditions can be found in \citet{Chen_et_al_2016}; \citet*{ota_kato_hara_2019, Zhang_Kato_Ruppert_2021}; however, each of these estimates requires four-times continuous differentiability of the density. Assumptions A4 and A5 concern the kernel function $k$ and the bandwidth parameter $h$ and are necessary for the previously stated benefits of the SQRE over the traditional QR estimator. A6 ensures uniqueness of the conditional mode and is also commonly used in deriving consistency of M-estimators, see Theorem 5.7 in \citet{vdVaart_1998}. Finally, Assumption A7 limits the possible values for the optimizer $\tau_x$, ensuring that the conditional modes are bounded away from the tails of the conditional distributions, uniformly on the covariate space.

\subsection{Consistency of the Convolution Mode Regression Estimator:}

The following lemma is a reinstatement of an inequality in \citet*{fernandes_guerre_horta_2021}. 
\begin{lemma}\label{FGH_prop1_lemma}
Under Assumption A1 to A5, it holds that
\begin{align}
\begin{split}
\left\Vert \widehat{D}_h(\tau) - D(\tau)\right\Vert
&= o(h^j) + O_P\left(\sqrt{{\log n}/{(nh)}}\right),
\end{split}
\end{align}
uniformly for \(\tau\in[\alpha,1-\alpha]\) and \(h\in[\underline{h}_n,\bar{h}_n]\).
\end{lemma}

\begin{proof}
See the proof of Proposition~1 in \citet*{fernandes_guerre_horta_2021}.
\end{proof}

Our next result is regarding the sample sparsity function and the fact that it converges to the population counterpart. 

\begin{lemma}\label{lemma_sparsity_consistency}
Under Assumptions A1 to A5, it holds that
\[\vert \widehat{s}_{x,h}(\tau) - s_x(\tau)\vert = o(h) + O_P\left(\sqrt{\log n/(nh)}\right)\]
uniformly over \(\tau\in [\alpha,1-\alpha]\),
\(x\in\mathcal{X}\) and \(h\in[\underline{h}_n,\bar{h}_n]\).
\end{lemma}
\begin{proof}
Write
\[-(\widehat{s}_{x,h}(\tau) - s_x(\tau))\, =\, 
x^\intercal[\widehat{D}_h(\tau)]^{-1}\bar{X}
\,
-
\,
x^\intercal [D(\tau)]^{-1}\mathbb{E}X.\]
Using \(\bar{X} = \bar{X} - \mathbb{E}X + \mathbb{E}X\) and rearranging, we have
\begin{align}
\begin{split}
-(\widehat{s}_{x,h}(\tau) - s_x(\tau)) &= x^\intercal[\widehat{D}_h(\tau)]^{-1}(\bar{X}-\mathbb{E}X) 
\\
&\qquad \ +\ x^\intercal \left\{ [\widehat{D}_h(\tau)]^{-1} - [D(\tau)]^{-1}\right\} (\mathbb{E}X)
\end{split}
\end{align}
Lemma~\ref{FGH_prop1_lemma} implies, by the local Lipschitz property of matrix inversion, that
\[
\left\Vert [\widehat{D}_h(\tau)]^{-1} - [D(\tau)]^{-1} \right\Vert = o(h) + O_P\left(\sqrt{\log n/(nh)}\right)
\]
uniformly in $\tau$ and $h$ as above. This together with
\[
\sup_{\tau,h} |\widehat{D}_h(\tau)| = O_P(1),
\quad
\bar{X} - \E X = O_P(1/\sqrt{n}),
\quad
\E X = O(1),
\quad
\sup_{x\in\mathcal{X}}\Vert x\Vert = O(1)
\]
tells us that
\begin{align} \label{s_rate}
\begin{split}
-(\widehat{s}_{x,h}(\tau) - s_x(\tau) )
&= x^\intercal O_p(1)O_p(1/\sqrt n)
\\
&\qquad+ x^\intercal\left( o(h) + O_P \left( \sqrt{{\log n}/{(nh)}} \right) \right)O(1)
\\
&= o(h) + O_P \left( \sqrt{{\log n}/{(nh)}}\right)
\end{split}
\end{align}
as stated.
\end{proof}

After showing that our sparsity functions are consistent, we prove that its maximizer $\widehat{\tau}_{x,h}$, in equation \eqref{taus}, is also consistent for $\tau_x$. 

\begin{theorem} \label{theorem_cons_tau}
Under Assumptions A1 to A7, it holds that
\begin{equation} \label{rate_tau}
    \widehat{\tau}_{x,h} = \tau_x +
    o(h^{1/2}) + O_P\Bigg( \Big( \frac{\log n}{nh} \Big)^{1/4}\Bigg)
\end{equation}
uniformly for \(x\in\mathcal{X}\) and \(h\in[\underline{h}_n,\bar{h}_n]\).
\end{theorem}

\begin{proof} The proof of Theorem (\ref{theorem_cons_tau}) consists of two parts: initially, it is proved that $\widehat{\tau}_{x,h}$ is consistent; then, in the second part of the proof, we calculate its rate of convergence.

\textbf{Part 1 (Theorem 1):} First, by the definition of $\widehat{\tau}_{x,h}$ and through Lemma~\ref{lemma_sparsity_consistency}, we have
\[
\widehat{s}_{x,h}(\widehat{\tau}_{x,h})\ge \widehat{s}_{x,h}(\tau_x)=s_x(\tau_x) + r_n,
\]
where $r_n = o(h) + O_P\left(\sqrt{\log n/(nh)}\right)$ uniformly over $\tau$, $x$ and $h$. Hence,
\begin{align}
\label{inequalities}
\begin{split}
s_x(\tau_x) - s_x(\widehat{\tau}_{x,h})
&\le
\widehat{s}_{x,h}(\widehat{\tau}_{x,h}) - s_x(\widehat{\tau}_{x,h}) - r_n
\\
&\le
|\widehat{s}_{x,h}(\widehat{\tau}_{x,h}) - s_x(\widehat{\tau}_{x,h})| + |r_n|
\\
&\le
\sup\nolimits_{\tau,x,h} |\widehat{s}_{x,h}(\tau) - s_x(\tau)| + \sup\nolimits_{\tau,x,h}|r_n|
\\
&= o(h) + O_P\left(\sqrt{\log n/(nh)}\right)
\end{split}
\end{align}
where the last equality follows again by Lemma~\ref{lemma_sparsity_consistency}.

Now, notice that compactness of $\mathcal{X}$, together with Assumptions A2, A6 and A7, ensure there exists an $\mathrm{x}\in\mathcal{X}$ such that
\[\sup_{x\in\mathcal{X}} \sup_{\tau:\ |\tau-\tau_x|\ge\epsilon} s_x(\tau) - s_x(\tau_x) = \sup_{\tau:\ |\tau-\tau_{\mathrm{x}}|\ge\epsilon} s_{\mathrm{x}}(\tau) - s_{\mathrm{x}}(\tau_{\mathrm{x}}) < 0\]
In view of this and using A6 once more, the following holds: for each \(\epsilon>0\) there exists an \(\eta>0\) such that the bound
\[s_x(\tau) \le s_x(\tau_x) - \eta\]
holds for all \(x\) in the support of \(X\) and all \(\tau\) with \(|\tau-\tau_x|\ge\epsilon\).

Using compactness of $\mathcal{X}\times [\underline{h}_n,\bar{h}_n]$ and letting \((\mathrm{x},\mathrm{h})\) attain the supremum \(\sup_{x,h} |\widehat{\tau}_{x,h} - \tau_x|\)  over \(\mathcal{X}\times [\underline{h}_n,\bar{h}_n]\), we have
\[\lbrace |\widehat{\tau}_{\mathrm{x},\mathrm{h}}-\tau_{\mathrm{x}}|\ge\epsilon \rbrace \subseteq \lbrace s_{\mathrm x}(\tau_{\mathrm x}) - s_{\mathrm x}(\widehat{\tau}_{\mathrm{x},\mathrm{h}}) \ge \eta \rbrace \subseteq \lbrace \sup\nolimits_{x,h} s_x(\tau_x) - s_x(\widehat{\tau}_{x,h}) \ge \eta\rbrace.\]
Thus, for any $\epsilon>0$,
\[
\Prob\lbrace\sup\nolimits_{x,h} |\widehat{\tau}_{x,h} - \tau_x|\ge\epsilon\rbrace
\le
\Prob\lbrace \sup\nolimits_{x,h} s_x(\tau_x) - s_x(\widehat{\tau}_{x,h}) \ge \eta\rbrace \to 0
\]
in view of \eqref{inequalities}.

\textbf{Part 2 (Theorem 1)} Recall the equations \eqref{sparsity-pop-explicit} and \eqref{sparsity-sample-explicit} with $D(\tau)$ and $\widehat{D}_h(\tau)$ defined as in \eqref{Ds_Rs-pop} and \eqref{Ds_Rs-sample}. The first derivative of $s_x(\tau)$ is as:
\begin{equation*}
    s_x^{(1)}(\tau):= \frac{\partial s_x(\tau)}{\partial \tau} = 
    x^\intercal \Bigl[ [D(\tau)]^{-1}D^{(1)}(\tau)[D(\tau)]^{-1} \mathbb{E}X 
\Bigr]
\end{equation*}
with $D^{(1)}(\tau)$ defined as $\partial D(\tau)/\partial\tau = \mathbb{E}[XX ^\intercal f^{(1)}(X ^\intercal \beta(\tau)|X) \cdot X^\intercal [D(\tau)]^{-1}\mathbb{E}X]$; the derivation of $D^{(1)}(\tau)$ is found in Appendix \ref{Appendix 1}, equation \eqref{d_first_deriv}. 

The first order condition $s_x^{(1)}(\tau_x)=0$ and its sample analog, $\widehat{s}^{(1)}(\widehat{\tau}_{x,h})=0$ yield:
\begin{align*}
\begin{split}
x^\intercal [D(\tau_x)]^{-1}D^{(1)}(\tau_x)[D(\tau_x)]^{-1} \mathbb{E}X  = 0 
\\
x^\intercal  [\widehat{D}_h(\widehat{\tau}_{x,h})]^{-1}\widehat{D}^{(1)}(\widehat{\tau}_{x,h})[D(\widehat{\tau}_{x,h})]^{-1} \bar{X} = 0
\end{split}
\end{align*}

Now, by a Taylor expansion with Lagrange remainder, we have
\begin{equation*}
    s_x(\widehat{\tau}_{x,h}) = s_x(\tau_x) + s^{(1)}_x(\tau_x)[\widehat{\tau}_{x,h} - \tau_x] + \frac{1}{2}s^{(2)}_x(\tau_x^*)[\widehat{\tau}_{x,h}-\tau_x]^2    
\end{equation*}
with $\tau^{*}_x$ as a point between $\widehat{\tau}_{x,h}$ and $\tau_x$. Assumptions A2, A6, and A7 ensure that $\tau\mapsto s_x(\tau)$ is strictly convex in a vicinity of $\tau_x$, so $\inf_\tau s^{(2)}(\tau) > 0$ in such a vicinity.  

Applying the first-order condition $s_x^{(1)}(\tau_x)=0$, we can rewrite the expansion as:
\begin{equation} \label{taylor_exp}
    s_x(\widehat{\tau}_{x,h}) = s_x(\tau_x) + \frac{1}{2}s^{(2)}_x(\tau_x^*)[\widehat{\tau}_{x,h}-\tau_x]^2,
\end{equation}
which leads to 
\begin{equation*}
    \vert \widehat{\tau}_{x,h} - \tau_x \vert = \sqrt{2}
\sqrt{\frac{\big|
s_x(\widehat{\tau}_{x,h}) - s_x(\tau_x) \big|
}{\vert s^{(2)}_x(\tau_x^*)\vert}}
=
\sqrt{\frac{o(h) + O_P\Big(\sqrt{\log n / (nh)}\Big)}
{O_P(1)}}
\end{equation*}

Using $\sqrt{a+b} \le \sqrt{a} + \sqrt{b}$, it yields our rate of convergence for $\widehat{\tau}_{x,h}$:
\begin{equation} 
    \vert \widehat{\tau}_{x,h} - \tau_x \vert =
    o(h^{1/2}) + O_P\Bigg( \Big( \frac{\log n}{nh} \Big)^{1/4} \Bigg)
\end{equation}
as stated.
\end{proof}

Now that the consistency for the quantile modes is proved and the rates of convergence are defined, we proceed to state the consistency for the estimator of the mode, $\widehat{m}_h(x)$. Our second theorem is constructed using previous results from this paper and from \citet*{fernandes_guerre_horta_2021}. 

\begin{theorem} \label{theorem_cons_m}
    If Assumptions A1 to A7 hold, then
\begin{equation} \label{m_rate}
        \widehat{m}_h(x) = m(x) + o(h^{1/2}) + O_P \Bigg( \Big( 
        \frac{\log n}{nh}\Big)^{1/4}\Bigg)
    \end{equation}
uniformly for \(x\in\mathcal{X}\) and \(h\in[\underline{h}_n,\bar{h}_n]\).

\end{theorem}
\begin{proof}
     From Theorem~\ref{theorem_cons_tau} we know that $\widehat{\tau}_{x,h} \xrightarrow{p} \tau_x$ at a rate of $o(h^{1/2}) + O_P ( {\log n}/{nh} )^{1/4}$. Also, from Theorems 1 and 2 in \citet*{fernandes_guerre_horta_2021} we have:
    \begin{equation} \label{rate_beta}
    \big\Vert \widehat{\beta}_h(\tau) - \beta(\tau) \big\Vert = O_P\left( \frac{1}{\sqrt{n}} + h^2 \right)    
    \end{equation}
    
Recalling that
    \begin{equation*}
        m(x) = x^\intercal \beta(\tau_x) 
        \quad \mathrm{and} \quad
        \widehat{m}_h(x) = x ^\intercal \widehat{\beta}_h(\widehat{\tau}_{x,h})
    \end{equation*}
we obtain
    \begin{align} \label{mhat-m}
    \begin{split}
        \vert \widehat{m}_h(x) - m(x) \vert &= 
        \Big| x ^\intercal \widehat{\beta}_h(\widehat{\tau}_{x,h}) -
        x^\intercal \beta(\tau_x) \Big| 
        \\
        &= \Big| x^\intercal \big( \widehat{\beta}_h(\widehat{\tau}_{x,h}) - \beta(\tau_x) \big) \Big| \le
        \Vert x \Vert \cdot \big\Vert \widehat{\beta}_h(\widehat{\tau}_{x,h}) - \beta(\tau_x) \big\Vert.
    \end{split}
    \end{align} 
    Given the differentiability condition (A2), we have that $\beta(\tau)$ is Lipschitz-continuous, thus, for some constant $C>0$, we have
    \begin{align} \label{lipschitz_cont_beta}
    \begin{split}
    \Vert \widehat{\beta}_h(\widehat{\tau}_{x,h}) - \beta(\tau_x)\Vert &\le 
    \Vert \widehat{\beta}_h(\widehat{\tau}_{x,h}) - \beta(\widehat{\tau}_{x,h}) \Vert + \Vert \beta(\widehat{\tau}_{x,h}) - \beta(\tau_x) \Vert 
    \\
    &\le(\sup\nolimits_\tau \Vert \widehat{\beta}_h(\tau) - \beta(\tau) \Vert)
    + C \Vert \widehat{\tau}_{x,h} - \tau_x\Vert 
    \\
    &\le O_p \left(\frac{1}{\sqrt{n}}+h^2\right) 
    + C \Biggl[
    o(h^{1/2}) + O_P\Bigg(\Big(\frac{\log n}{nh}\Big)^{1/4}\Bigg)
    \Biggr],
    \end{split}    
    \end{align}
which yields \eqref{m_rate}.
\end{proof}

\begin{remark}
Denote our rate of convergence from Theorem~\ref{theorem_cons_m} as $R_{CMR}$ and the rate from the estimator proposed by \citet*{Zhang_Kato_Ruppert_2021} as $R_{ZKR}$,
\begin{align}
\begin{split} \label{conv_rates}
    R_{CMR} &= O_P \Big( n^{-1/4}h^{-1/4}(\log n)^{1/4} \Big) + o(h^{1/2})
    \\
    R_{ZKR} &= O_P \Big( n^{-1/2}h^{-3/2} (\log n)^{1/2} + h^2 \Big)
\end{split}    
\end{align}
Neither rate is dimension dependable, thus free from the "curse of dimensionality"; under certain conditions on $h$, our rate, $R_{CMR}$, is marginally slower than $R_{ZKR}$. Apart from the presence of the deterministic term, both rates are similar, and the difference lies on the selection of the bandwidth parameter, $h$. Given a certain bandwidth, $R_{ZKR}$ can achieve, at best, a rate of $(n\log n)^{-2/7}$, similar to the rate of \citet{Kemp_Silva_2012} and faster than the rate in \citet*{ota_kato_hara_2019}. Despite the differences in convergence rates, we attain uniformity with respect not only to the design points of the covariates, but also to the bandwidth.

\end{remark}

\begin{remark} 
Rewrite $R_{ZKR}$ in \eqref{conv_rates} as
\begin{equation}
O_P\left( \left[\frac{\log n}{n h^3}\right]^{1/2} + h^2\right).
\end{equation}
The ratio
\begin{equation}\label{eq:rates-ratio}
\frac{(\log n/(nh))^{1/4}}{(\log n/(nh^3))^{1/2}} = \frac{n h^5}{\log n}
\end{equation}
diverges to infinity under Assumption (viii) in \citet*{Zhang_Kato_Ruppert_2021}, so under this assumption our estimator cannot achieve $O_P$ rates faster than theirs. Nevertheless, under our weaker Assumption A5, we can make \eqref{eq:rates-ratio} go to zero, for example by taking $h = (n/\log n)^{-1/5}b$ with $b\to0$ and $b^3(n/\log n)^{-2/5}\to\infty$. However, this particular choice for the bandwidth is not contemplated due to \citeauthor*{Zhang_Kato_Ruppert_2021}'s \citeyearpar{Zhang_Kato_Ruppert_2021} assumptions, but is enabled by our condition A5. 

\end{remark}

\begin{remark}
Importantly, our estimator $\widehat{m}_h(x)$ attains the rate in Theorem~\ref{theorem_cons_m} uniformly both in $x$ and $h$; on the other hand, the representation in Proposition~1 of \citet*{Zhang_Kato_Ruppert_2021} is not uniform for the bandwidth. Obtaining uniformity in $h$ can be useful for 3 types of bandwidth choices: (i) data-driven bandwidth choices, as in \citet*{fernandes_guerre_horta_2021}; (ii) adaptive bandwidth choices, such as the ones of \citet{terrell_scott_1992, Lepski_etal_1997}; and (iii) choices robust to bandwidth-snooping, as in \citet{armstrong_kolesar_2018}.

\end{remark}


\section{Monte Carlo Study} \label{Simulation}
In this section we employ a preliminary\footnote{For a more comprehensive Monte Carlo Study, considering different types of error terms, we refer the reader to the work of \citet{ongaratto_horta_2021}.} version of a Monte Carlo study in order to illustrate if our estimator converges to a normal distribution in finite samples. All the analyses in this section were developed using \textbf{\textsf{R}} language, version 4.3.1 \citep{R_software} via the software \textbf{\textsf{RStudio}} \citep{R_studio}. The estimation of $\widehat{\beta}_h(\tau)$, as in \eqref{SQRE}, was done via the package \textbf{\textsf{conquer}} from \citet{conquer}.

\subsection{Simulation Design}
In our simulation, we consider a type of heteroscedastic error distribution, similar to \citet*{fernandes_guerre_horta_2021}. Our pseudorandom data is generated from the model $Y = X^\intercal \beta + \epsilon$, where $X = (1, \Tilde{X})^\intercal$, with $\Tilde{X} \sim U(1,5)$, $\beta = (1,1)$ and the conditional heteroscedastic error term $\epsilon = \frac{1}{2} (1+\Tilde{X})Z$, with $Z$ following a Skew Normal distribution with zero mean, unit variance and parameter $\alpha=2$. As for the design point of the covariates, we condition $X=x$ with $x=(1, 3)^\intercal$. The simulations require $\tau$ values, so we consider a grid, $\mathcal{T} = \{ 0.01, 0.02, ..., 0.98, 0.99\}$. We contemplated 3 different samples sizes, $n = \{ 100, 250, 500\}$, each for the number of replications, $it=\{100, 1000\}$.

\subsection{Bandwidth Selection}
It is necessary to choose the bandwidth parameter $h$ in order to estimate $\widehat{m}_h(x)$ as in \eqref{mode_asQ}. In line with \citet*{fernandes_guerre_horta_2021}, we employ the \textit{Silverman's rule-of-thumb bandwidth} \citep{silverman_1986}. We consider a $\tau$-dependent, plug-in rule-of-thumb bandwidth which can be expressed as:
\begin{equation} \label{hrot}
    h_{ROT}(\tau) = \frac{1.06}{\sqrt[5]{n}} \widehat{S}(\tau), \quad \tau \in \mathcal{T}
\end{equation}
where, for each quantile level $\tau \in \mathcal{T}$, there is a different value of $\widehat{S}(\tau)$, which is computed according to the following algorithm:
\begin{enumerate}
    \item Calculate the residuals for the canonical quantile regression estimator $\beta(\tau)$ as in $e_i(\tau):= Y_i - X_i^\intercal (\widehat{\beta}(\tau)$ for $i \in \{1,...,n\}$;
    \item The sample interquartile range $iq(\tau)$ is calculated, corresponding to the residuals $\widehat{e}_1(\tau),...,\widehat{e}_n(\tau)$;
    \item The sample standard deviation $\widehat{\sigma}(\tau)$, from $\widehat{e}_1(\tau),...,\widehat{e}_n(\tau)$, is calculated;
    \item Lastly, $\widehat{S}$ is given by the $\min{0.7199528 \times iq(\tau); \sigma(\tau)}$.
\end{enumerate}

Since the bandwidth $h_{ROT}$ is data-driven, it depends on the generated sample, thus it can vary with each replication. Also, the bandwidth is $\tau$-dependent, therefore, there is a different smoothing parameter for each value of $\tau \in \mathcal{T}$.

\subsection{Limiting Distribution of the Estimator}
Since no distribution is achieved in our asymptotic theory, we rely on simulations in finite sample sizes to analyze the behavior of the \textit{Convolution Mode Regression Estimator}, specifically, its limiting distribution. In order to deduce if our proposed estimator is normally distributed, we analyze the quantile-quantile (Q-Q) plots for the estimated $\widehat{m}(x)$ against the theoretical quantiles of a Normal Distribution.

\vfill

\begin{figure}
\caption{Q-Q plot for $\widehat{m}(x)_{h=ROT}$ sample quantiles (black) against Normal quantiles (red)}
\centering
\includegraphics[width=0.75\textwidth]{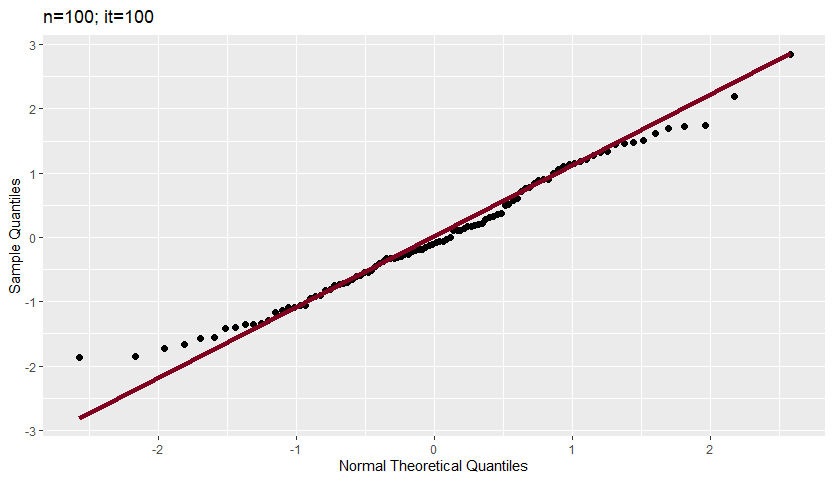}
\label{n100_it100}

\vspace{0.5cm}
\caption{Q-Q plot for $\widehat{m}(x)_{h=ROT}$ sample quantiles (black) against Normal quantiles (red)}
\centering
\includegraphics[width=0.75\textwidth]{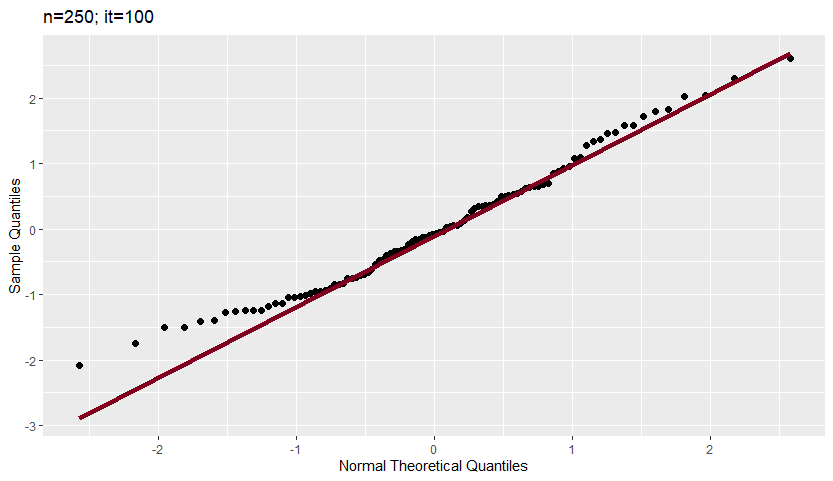}
\label{n250_it100}

\vspace{0.5cm}
\caption{Q-Q plot for $\widehat{m}(x)_{h=ROT}$ sample quantiles (black) against Normal quantiles (red)}
\centering
\includegraphics[width=0.75\textwidth]{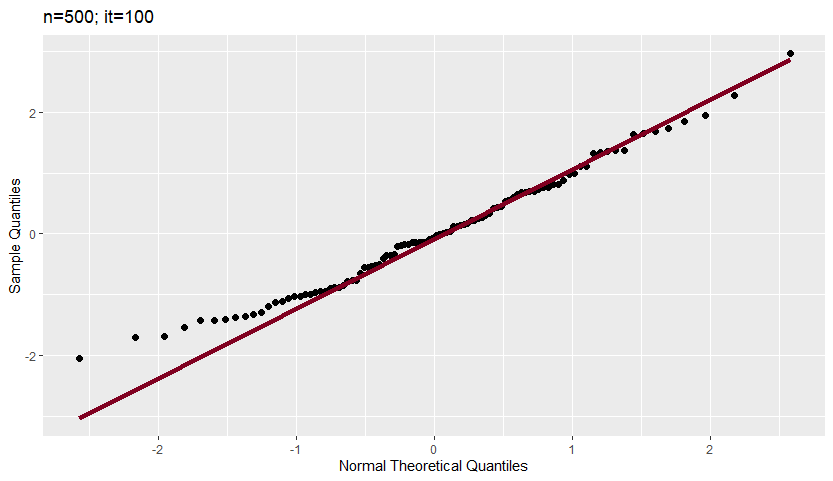}
\label{n500_it100}
\end{figure}

\begin{figure}
\caption{Q-Q plot for $\widehat{m}(x)_{h=ROT}$ sample quantiles (black) against Normal quantiles (red)}
\centering
\includegraphics[width=0.75\textwidth]{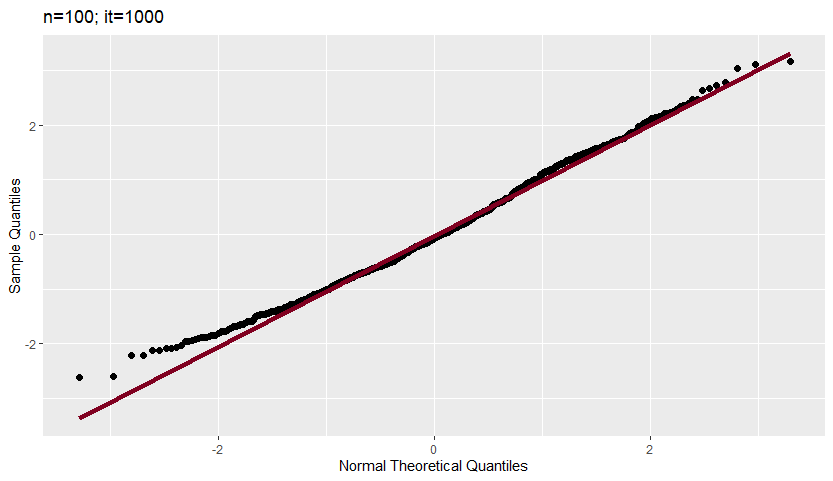}
\label{n100_it1000}

\vspace{0.5cm}
\caption{Q-Q plot for $\widehat{m}(x)_{h=ROT}$ sample quantiles (black) against Normal quantiles (red)}
\centering
\includegraphics[width=0.75\textwidth]{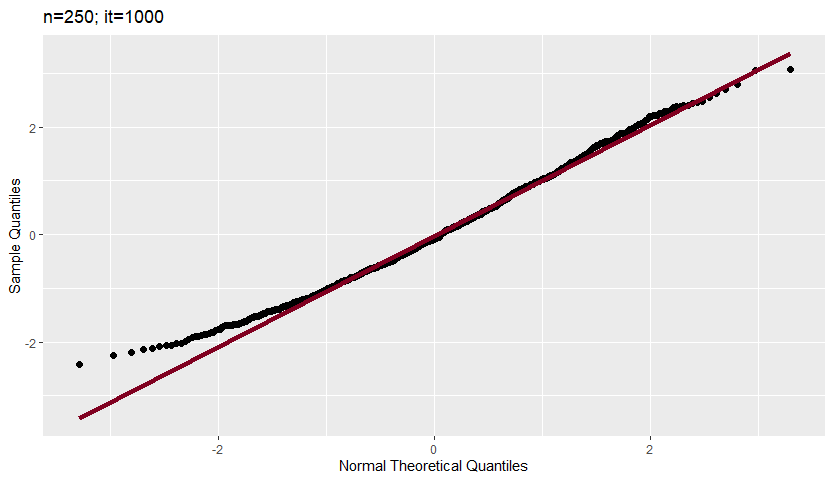}
\label{n250_it1000}

\vspace{0.5cm}
\caption{Q-Q plot for $\widehat{m}(x)_{h=ROT}$ sample quantiles (black) against Normal quantiles (red)}
\centering
\includegraphics[width=0.75\textwidth]{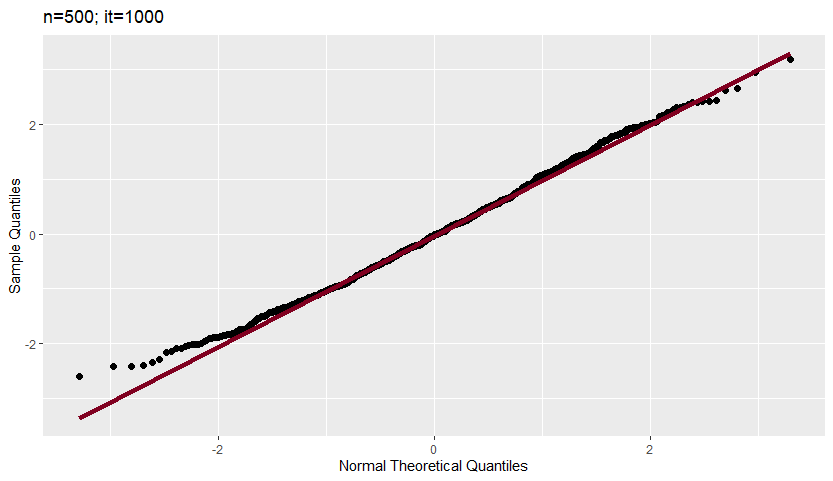}
\label{n500_it1000}
\end{figure}

According to Figures \ref{n100_it100}-\ref{n500_it100} there is some evidence that our estimator is normally distributed, for the referred sample sizes. The generated sample points show some adherence to the 45° degree line, which represents the theoretical quantiles of a Normal Distribution. The “stickiness” of the points to the line seems to be weaker on the lower tail for all sample sizes with 100 replications. When we increase the number of replications of the simulation for $it=1000$ (Figures \ref{n100_it1000} and \ref{n250_it1000}), we get similar results to the previous ones, the generated quantile points of $\widehat{m}_{h_{ROT}}(x)$ overlap the red line, except in the lower limits of the samples. The last simulated example, in Figure \ref{n500_it1000} ($n=500$ and $it=1000$), shows a better correspondence of the simulated points of the lower tail to the theoretical pathway. In short, the different simulated scenarios suggest that for finite samples, our estimator gets closer to being normally distributed. 



\section{Concluding Remarks} \label{Concluding Remarks}

In the present paper we developed a novel estimator for the conditional mode $\widehat{m}_h(x)$, called \textit{Convolution Mode Regression}, which is based on inverting a convolution-type smoothed quantile regression. The idea of achieving the conditional mode via quantile regression is not groundbreaking, since it has been done previously \citep*{ota_kato_hara_2019, Zhang_Kato_Ruppert_2021}. Despite that, it presents advantages regarding the two main problems with mode regression, slow convergence in nonparametric settings and nonconvex optimization in linear environments. Since the mode is the maximizer of the density and the density can be retrieved from the qdf, our estimation strategy relies on an intermediate step. Specifically, to estimate $\widehat{m}_h(x)$, we first estimate the conditional quantile function $Q(\cdot|x)$, followed by the conditional quantile mode of $Y$ given $X=x$, denoted $\tau_x$. The qdf estimator we rely on, from \citet*{fernandes_guerre_horta_2021}, makes use of the linear structure of QR and also converges to the real pdf curve.

Differently from the existing work of \citet*{Zhang_Kato_Ruppert_2021}, who initially estimate the quantile regression then smooth it through a kernel, our approach relies on "smooth then estimate". We develop asymptotic consistency for our estimator, obtaining some differences in convergence rates, as compared to \citet*{Zhang_Kato_Ruppert_2021}. Apart from the initial smoothing, the main differentiation of our model from the authors' is in the bandwidth selection premise, since our assumption for the choice of $h$ is less restrictive. Furthermore,  the uniformity of $\widehat{m}_h(x)$ with respect to $h$, something not obtained in \citet*{Zhang_Kato_Ruppert_2021}, makes our model an interesting choice when the bandwidth selection is data-driven or adaptive. Our preliminary simulations show, that for finite sample sizes, our estimator behaves similar to a Normal Distribution.

Further work related to present research can take many directions. In what we assess as more important, the continuation of the asymptotic properties, namely, the limiting distribution of the estimator. Furthermore, to refine the Monte Carlo Study, more in line with \citet{ongaratto_horta_2021}, in order to compare the performance of the \textit{Convolution Mode Regression Estimator} against \citeauthor*{Zhang_Kato_Ruppert_2021}'s. In accordance to that, implementations of our estimator in previous applied econometric  and predictive problems can be carried out; besides to a generalization of the present framework for time series, focused on forecasting models for asymmetric data.

\appendix

\section{Appendix}

\subsection{First and Second Derivatives of $\widehat{R}_h(b;\tau)$}
\label{Appendix 0}

From \citet*{fernandes_guerre_horta_2021}, the first and second derivatives of the smoothed sample objective function, $\widehat{R}_h(b;\tau)$, with respect to $b$, are, respectively:

\begin{align} \label{derivates_of_Rh}
\begin{split}
\widehat{R}_h^{(1)}(b;\tau) &= \frac{1}{n} \sum_{i=1}^{n}X_i \Biggl[ K \Biggl( \frac{X_i^\intercal b - Y_i}{h} \Biggr) - \tau 
\Biggr]
\\
\widehat{R}_h^{(2)}(b;\tau) &= \frac{1}{n}\sum_{i=1}^n X_i X_i^\intercal k_h(X_i^\intercal b - Y_i)
\end{split}
\end{align}
with $K(t) := \int_{-\infty}^{t} k(v)dv$. 

\subsection{Derivation of $D^{(1)}(\tau)$} \label{Appendix 1}

Recalling the definition of $D(\tau)$:
\begin{equation*}
    D(\tau):= R^{(2)}(\beta(\tau);\tau) = \E [XX^\intercal f(X^\intercal \beta(\tau)|X)]
\end{equation*}
The first order differentiation is expressed as:
\begin{align} \label{d_first_deriv}
\begin{split}
D^{(1)}(\tau) &:= \frac{\partial}{\partial \tau} \mathbb{E}[ XX^\intercal f(X^\intercal \beta(\tau)|X)]
\\
&= \mathbb{E}[XX ^\intercal f^{(1)}(X ^\intercal \beta(\tau)|X) \cdot X^\intercal \beta^{(1)}(\tau)]
\\
&= \mathbb{E}[XX ^\intercal f^{(1)}(X ^\intercal \beta(\tau)|X) \cdot X^\intercal [D(\tau)]^{-1}\mathbb{E}X]
\\
D^{(1)}(\tau) &= \mathbb{E}\Bigl[XX ^\intercal f^{(1)}(X ^\intercal \beta(\tau)|X) \cdot X^\intercal \Bigl\{ \mathbb{E}[XX^\intercal f(X^\intercal \beta(\tau)|X)] 
\Bigr\}^{-1}\mathbb{E}X\Bigr]
\end{split}
\end{align}

\subsection{Derivation of $s_x(\tau)$} \label{Appendix 2}
Recalling the definition of the population sparsity function:
\begin{equation*}
    s_x(\tau) = -\frac{\partial}{\partial\tau} x^\intercal\beta(\tau) =  -x^\intercal \underbrace{[D(\tau)]^{-1} \mathbb{E}X}_{\beta^{(1)}(\tau)}
\end{equation*}
To calculate the first derivative of $s_x(\tau)$, we use the definition of $\beta(\tau)$ as in the previous equation:
\begin{equation} \label{s_beta}
    s^{(1)}_x(\tau) : = \frac{\partial s_x(\tau)}{\partial \tau} = \frac{\partial}{\partial \tau} -x^\intercal \beta^{(1)}(\tau) = -x^\intercal \beta^{(2)}(\tau)
\end{equation}
Now, computing $\beta^{(2)}(\tau)$:
\begin{align} \label{beta2_def}
\begin{split}
\beta^{(2)}(\tau) &:= \frac{\partial \beta^{(1)}(\tau)}{\partial \tau} = \frac{\partial}{\partial \tau} [D(\tau)]^{-1} \mathbb{E}X
\\
&= -D^{(1)}(\tau)[D(\tau)]^{-2} \mathbb{E}X
\\
&=-D^{(1)}(\tau) [D(\tau)]^{-1}\beta^{(1)}(\tau)
\\
&= -[D(\tau)]^{-1}D^{(1)}(\tau)[D(\tau)]^{-1} \mathbb{E}X
\end{split}    
\end{align}
Applying the result in \eqref{beta2_def} to equation \eqref{s_beta} we get $s_x^{(1)}(\tau)$:
\begin{equation} \label{s_first_deriv}
    s_x^{(1)}(\tau) = x^\intercal \Bigl[ [D(\tau)]^{-1}D^{(1)}(\tau)[D(\tau)]^{-1} \mathbb{E}X 
\Bigr]
\end{equation}
with $D^{(1)}(\tau)$ defined as in equation \eqref{d_first_deriv}.

The second derivative of $s_x(\tau)$ is required in the Taylor Expansion \eqref{taylor_exp}, so we compute $s^{(2)}(\tau)$ as follows:
\begin{align}
\begin{split} \label{s_second_deriv}
s^{(2)}(\tau) &:= \frac{\partial s^{(1)}(\tau)}{\partial \tau} = \frac{\partial}{\partial \tau} 
x^\intercal [D(\tau)]^{-1} D^{(1)}(\tau) [D(\tau)]^{-1} \mathbb{E}X
\\
s^{(2)}(\tau) &= x^\intercal \Biggl[ 
\Bigl( 
[D(\tau)]^{-1}D^{(2)} - 2[D^{(1)}(\tau)]^2
\Bigr)
[D(\tau)]^{-3}
\Biggr] 
\mathbb{E}X
\end{split}
\end{align}

\bibliographystyle{chicago}
\bibliography{refs}

\end{document}